\documentclass[1p,final,12pt,sort,times]{elsarticle}
\journal{Finite Fields and Their Applications}
\usepackage{amsmath,amssymb}
\addtocounter{MaxMatrixCols}{1} 

\renewcommand{\tilde}[1]{\widetilde{#1}}

\newcommand{\vr}[2]{{\mathbb{F}}_{#1}^{#2}}                   
\newcommand{\bm}[1]{\mbox{\boldmath $ #1$} }         
\newcommand{\Supp}{\textrm{Supp}}
\newcommand{\lm}[1]{{\mathbb F}_{#1}}

\newtheorem{theorem}{Theorem}
\newtheorem{lemma}[theorem]{Lemma}
\newtheorem{proposition}[theorem]{Proposition}

\newdefinition{definition}[theorem]{Definition}
\newdefinition{remark}[theorem]{Remark}
\newdefinition{example}{Example}

\usepackage[T1]{fontenc}
\usepackage[pdfstartview={FitBH -32768},pdfauthor={Olav Geil, Ryutaroh Matsumoto and Diego Ruano},pdftitle={Feng-Rao decoding of primary codes},pdfkeywords={decoding, Feng-Rao bound, generalized Hamming weight, minimum distance, order domain, well-behaving pair}]{hyperref}
\newcommand{\widebar}[1]{\overline{#1}}

\newproof{proof}{Proof}

\begin{document}
\begin{frontmatter}

\title{Feng-Rao decoding of primary codes} 
 
 \author[geilandruano]{Olav Geil}
 \ead[url]{http://people.math.aau.dk/~olav}
 
 \author[matsumoto]{Ryutaroh Matsumoto}
 \ead[url]{http://www.rmatsumoto.org/research.html}
 
 \author[geilandruano]{Diego Ruano}
 \ead[url]{http://people.math.aau.dk/~diego}

 \address[geilandruano]{Department of Mathematical Sciences, 
    Aalborg University, Denmark}
 \address[matsumoto]{Department of Communications and Integrated Systems,
    Tokyo Institute of Technology, Japan}
 \begin{abstract}
We show that the Feng-Rao bound for dual codes
and a similar bound by Andersen and Geil \cite{AG} for primary
codes are consequences of each other. This implies that the Feng-Rao decoding algorithm
can be applied to decode primary codes up to half their designed
minimum distance. The technique
 applies to any linear code for which information on well-behaving
pairs is available. Consequently we are able to decode efficiently a large
class of codes for which no non-trivial decoding algorithm was
previously known. Among those are important families of multivariate polynomial codes. Matsumoto and
Miura in~\cite{MM} (See also~\cite{petertom}) derived from the Feng-Rao bound a bound for primary
one-point algebraic geometric codes and showed how to decode up to
what is guaranteed by their bound. The exposition in~\cite{MM} requires the use
of differentials which was not needed in \cite{AG}. Nevertheless we
demonstrate a very strong connection between Matsumoto and Miura's
bound and 
Andersen and Geil's bound when applied to primary one-point algebraic
geometric codes.
\end{abstract}
 
\end{frontmatter}
\thispagestyle{plain}
 
\section{Introduction}\label{secintro}\label{sec1}
Originally the celebrated Feng-Rao bound  was stated \cite{FR24,FR1,FR2} in
the language of 
affine variety codes \cite{lax}. Later
H{\o}holdt, Pellikaan, and van Lint \cite{handbook} introduced the
concept of 
order domains and order functions to facilitate the use of the bound, and for those
structures it was renamed the order bound. The success of the order
bound formulation comes from the fact that order domain codes include
the important family of duals of one-point algebraic geometric 
codes as well as the generalization of such codes to higher
transcendence degree algebraic structures. A completely different
point of view was to
formulate the Feng-Rao bound in the setting of general linear
codes~\cite{early,MM,miura2,miura1,handbook}. In this setting having no supporting algebra, a grading of
${\mathbb{F}}_q^n$ is assumed. This simply corresponds to defining an
indexed basis. The componentwise product then plays the role that should
otherwise be  played by 
multiplication in the supporting algebra. It should be stressed that the linear code formulation
is the most general in the sense that the other formulations mentioned
above can be deduced from that. So results derived in the linear code
description can be easily translated into the situation where we have
some kind of a supporting algebra. This of course in particular holds
for the decoding method to be discussed in the present paper.

The strength of the Feng-Rao bound --- besides the fact that it improves
on previous bounds such as the Goppa bound --- is that it enables an
improved code construction \cite{FR2}. Furthermore, it comes with a decoding
algorithm that applies to any dual code, improved or not. This decoding
algorithm can be formulated in any of the three settings: affine
variety codes, order domain codes, and general linear codes. 

Building on~\cite{normtrace,hyperbolic,shibuya} Andersen and Geil
in~\cite{AG} introduced a bound on the minimum distance of primary
codes. This bound was later slightly generalized and enhanced in~\cite{geithom}, but we
shall refer also this version as Andersen and Geil's bound. The bound has the same flavor as the Feng-Rao
bound. In particular it also enables an
improved code construction. The exposition in~\cite{AG}
starts by treating the general linear code set-up. It is then simply a 
matter of translation to reformulate the bound in the setting of order
domain codes and affine variety codes~\cite{AG,BookAG}. In particular
an improvement to the Goppa bound for primary one-point algebraic
geometric codes is given in~\cite[Th.~33 and Pr.~37]{AG}. Recent
papers~\cite{lee,LBO} show how to decode a certain class of one-point algebraic geometric codes up
to half the value of Andersen and Geil's bound. (See~\cite[Prop.~6]{GMRsubmitted} for a
proof that the error-correcting capability is actually that of Andersen
and Geil's bound). A generalization of
the previous algorithm \cite{LBO} for decoding general primary one-point
algebraic geometric codes has been given in
\cite{GMRisit,GMRsubmitted}, furthermore it can decode beyond that number.

Although the two bounds are of a similar flavor, to the best of our
knowledge till now
no general correspondence between the Feng-Rao bound and 
Andersen and Geil's bound has been established. The following is known about
the correspondence: Firstly, Shibuya and Sakaniwa in~\cite{shibuya} derive a
bound on the minimum distance of primary codes. This bound relies on
the Feng-Rao bound for generalized Hamming weights. As demonstrated
in~\cite[Sec.~5]{AG} one can in a certain sense view Andersen and
Geil's 
bound as an improvement to Shibuya and Sakaniwa's bound. 
Secondly, for the case
of isometry-dual one-point algebraic geometric codes it was shown
in~\cite{GMRT} that the Feng-Rao bound and Andersen and Geil's bound
produce the same result. This is in contrast with the general case of
one-point algebraic geometric codes where the two bounds may produce
completely different values~\cite[Ex.\
51]{AG}. In the light of Section~\ref{secfivefive} below,
 \cite[Sec.~4]{MM} constitutes another example of the two
bounds producing different values. Finally, a result in a different direction was established
in~\cite{GMRT} where it was shown that for one-point algebraic geometric codes one can view 
 Andersen and Geil's bound as a consequence of the Beelen bound \cite{Beelen} for
more point codes and thereby also as a consequence of the Duursma-Kirov-Park
bound \cite{DK,DKP} for such codes. However, it seems prohibitively difficult to prove
the equality between the error correction capability of \cite{LBO,lee}
and half the bounds in \cite{Beelen,DK,DKP}, while
we proved  in just a few lines \cite[Prop.\ 6]{GMRsubmitted}
the equality between \cite{LBO,lee} and half the bound in \cite{AG}.
This demonstrates that Andersen and Geil's bound \cite{AG}
is much more convenient than \cite{Beelen,DK,DKP} in some cases,
though the former \cite{AG} is implied by the latter \cite{Beelen,DK,DKP}.

The translations via generalized Hamming weights and via 
 more point
codes do not seem to suggest a simple connection between the
Feng-Rao bound and Andersen and Geil's bound for minimum
distance. Nevertheless, we shall
demonstrate that such a connection does indeed exist.
As a consequence, we can see that the error correction capability of the
recently proposed decoding algorithms \cite{LBO,lee}
is equal to the Feng-Rao decoding algorithm for primary codes \cite{MM}.

The above connection is of academic interest itself. But maybe more
importantly, it enables us to decode primary codes up to what is
guaranteed by Andersen and Geil's bound. As shall be demonstrated in
the present paper it suffices to derive a particular dual description
of the codes by means of linear algebra, and then to apply the three-bases
 generalization of the Feng-Rao decoding algorithm
in~\cite[Sec.\ 4.3]{handbook} and \cite[Sec.\ 2]{MM}, while a similar
generalization appeared much earlier in \cite{early}. The technique
applies to a large
variety of codes for which no efficient decoding algorithms are
known. This includes important families of multivariate polynomial
codes often considered by theoretical computer
scientists.
Another implication of the above mentioned
connection is that it becomes clear that Andersen and Geil's bound is
in some sense a generalization of Matsumoto and Miura's bound for
primary one-point algebraic geometric
codes~\cite[Secs.\ 3 \& 4]{MM}. This also implies that the decoding method
of the present paper can be viewed as a generalization of the decoding
method for primary one-point algebraic geometric codes in~\cite{MM}. It
should be mentioned that another generalization of Matsumoto and
Miura's bound and decoding method is given by Beelen and H{\o}holdt in~\cite{petertom}
where more point codes $C_{\mathcal{L}}(D,G)$ are treated.

The present paper starts by treating in Sections~\ref{sectwo} and
\ref{secthree} the general case of linear codes. Section~\ref{sectwo}
describes the state of the art and Section~\ref{secthree}  establishes the connection
between the two bounds. In Section~\ref{secfour} we briefly discuss
how to use the results from Section~\ref{sectwo} and \ref{secthree}
when a supporting algebra is given. A couple of examples illustrate
the idea. The new decoding method for primary codes is then treated in
Section~\ref{secfive}. In Section~\ref{secfivefive} we investigate the
connection between Andersen and Geil's bound when applied to primary
one-point algebraic geometric codes and the bound by Matsumoto and
Miura for similar codes. The correspondence in Section~\ref{secthree}
has implications for the estimation of generalized Hamming weights. We
briefly comment on this fact in Section~\ref{secsix}.

As a consequence of our findings in the present paper we suggest
that the Feng-Rao bound is in the future called ``the Feng-Rao bound for
dual codes'' and that Andersen and Geil's bound is called ``the Feng-Rao
bound for primary codes''. Similarly, we suggest that the order bound is
in the future called ``the order bound for dual codes'' whereas the bound
by Andersen and Geil for order domain codes is named ``the order bound
for primary codes''. We shall stick to this naming throughout the
remaining part of the paper.

\section{The general linear code formulation}\label{sectwo}
Let ${\mathcal{B}}=\{ {\bm{b}}_1, \ldots ,  {\bm{b}}_n\}$ be a basis for 
${\mathbb{F}}_q^n$ as a vector space over ${\mathbb{F}}_q$. Consider a non-empty set $I\subseteq \{1, 2, \ldots
, n \}$. We shall study the code  $$C({\mathcal{B}},I)={\textrm{span}}_{\lm{q}}\{\bm{b}_i \mid
i \in I\}$$ and its dual, which we denote by 
$C^{\perp}({\mathcal{B}},I)$.\\
Let  $L_{-1}=\emptyset$, $L_0=\{ \bm{0}\}$, and define for
$l=1,\ldots , n$, $L_l={\textrm{span}}_{\lm{q}}\{\bm{b}_1, \ldots ,\bm{b}_{l}\}$.
We have $$\emptyset=L_{-1} \subsetneq L_0
\subsetneq \cdots \subsetneq L_n=\vr{q}{n}.$$ Hence, the
following definition makes sense.

\begin{definition}\label{defone}
Define $\widebar{\rho}_{\mathcal{B}}: \vr{q}{n} \rightarrow \{0, 1, \ldots ,n\}$ by
$\widebar{\rho}_{\mathcal{B}}(\bm{v})=l$ if $\bm{v} \in L_l \backslash L_{l-1}$.
\end{definition}

We equip ${\mathbb{F}}_q^n$ with a second binary operation namely the
component-wise product
$$(u_1, \ldots , u_n) \ast (v_1, \ldots , v_n)=(u_1v_1, \ldots , u_nv_n).$$ 
With the above in hand we can introduce the concept of well-behaving pairs
which plays a fundamental role in the Feng-Rao bound for dual codes as
well as in the Feng-Rao bound for primary codes.

\begin{definition}
Consider two bases ${\mathcal{B}}=\{ {\bm{b}}_1, \ldots ,
{\bm{b}}_n\}$ and ${\mathcal{U}}=\{ \bm{u}_1,\ldots \bm{u}_n\}$ for $
\vr{q}{n}$ as vector space over ${\mathbb{F}}_q$ (we may or may not have ${\mathcal{B}}={\mathcal{U}}$).\\
An ordered pair $(i,j) \in \{1, 2, \ldots , n\}\times \{1, 2, \ldots , n\} $ is said to be well-behaving
(WB) with respect to $({\mathcal{B}},{\mathcal{U}})$ if $$\widebar{\rho}_{\mathcal{B}}(\bm{b}_u \ast \bm{u}_v)
< \widebar{\rho}_{\mathcal{B}}(\bm{b}_i \ast \bm{u}_j)$$ holds 
for all $u$ and $v$ with $1 \leq u \leq i, 1 \leq v \leq j$ and $(u,v)
\neq (i, j)$. \\
An ordered pair $(i,j) \in \{1, 2, \ldots , n\}\times \{1, 2, \ldots ,
n\} $ is said to be weakly well-behaving
(WWB) with respect to $({\mathcal{B}},{\mathcal{U}})$ if 
$$
\begin{array}{c}\widebar{\rho}_{\mathcal{B}}(\bm{b}_u \ast \bm{u}_j)
< \widebar{\rho}_{\mathcal{B}}(\bm{b}_i \ast \bm{u}_j),\\
\widebar{\rho}_{\mathcal{B}}(\bm{b}_i \ast \bm{u}_v)
< \widebar{\rho}_{\mathcal{B}}(\bm{b}_i \ast \bm{u}_j)
\end{array}
$$ 
hold
for all $u < i$ and $v<j$.\\
Even less restrictively an ordered pair $(i,j) \in \{1,2,\ldots , n\} \times \{1, 2, \ldots , n\}$ is said
to be one-way well-behaving (OWB) with respect to $({\mathcal{B}},{\mathcal{U}})$ 
if $$\widebar{\rho}_{\mathcal{B}}(\bm{b}_u \ast \bm{u}_j) <
\widebar{\rho}_{\mathcal{B}}(\bm{b}_i \ast \bm{u}_j)$$ holds 
for $u<i$.
\end{definition}

\begin{remark}\label{remrom}
Clearly, WB implies WWB which again implies OWB.
\end{remark}

The Feng-Rao bound for dual codes and the Feng-Rao bound for primary codes are about counting well-behaving pairs satisfying
certain criteria. Assume in the following that ${\mathcal{B}}$ and ${\mathcal{U}}$ are
fixed. To introduce the 
Feng-Rao bound for dual codes we define for 
$l=1,2, \ldots , n$ 
\begin{eqnarray}
\widebar{\mu}^{\textrm{WB}}_{({\mathcal{B}},{\mathcal{U}})}(l)&=&\sharp \{ i \in \{1, 2, \ldots , n\}  \mid  \widebar{\rho}(\bm{b}_i \ast \bm{u}_j)=l {\mbox{ \ for some \ }}
\bm{u}_j \in {\mathcal{U}} \nonumber \\
&&   {\mbox{\ \hspace{6.5cm} with \ }}  (i,j) {\mbox{\   WB
}}\}, \nonumber \\
\widebar{\mu}^{\textrm{OWB}}_{({\mathcal{B}},{\mathcal{U}})}(l)&=&\sharp \{ i \in \{1,2,\ldots , n\}  \mid  \widebar{\rho}(\bm{b}_i \ast \bm{u}_j)=l {\mbox{ \ for some \ }}
\bm{u}_j \in {\mathcal{U}}  \nonumber \\
&& {\mbox{\  \hspace{6.5cm} with \ }}  (i,j) {\mbox{\   OWB
}}\}. \nonumber 
\end{eqnarray}
We stress that our definition of $\widebar{\mu}^{\textrm{WB}}_{({\mathcal{B}},{\mathcal{U}})}$ is equivalent to the
following form usually found in the literature:
\begin{eqnarray}
\widebar{\mu}^{\textrm{WB}}_{({\mathcal{B}},{\mathcal{U}})}(l)&=&\sharp\{(i,j) \in \{1,2,\ldots ,
n\}\times \{1, 2, \ldots , n\}  \mid \widebar{\rho}(\bm{b}_i \ast
\bm{u}_j)=l \nonumber \\
&&{\mbox{ \ \hspace{6.0cm}  \ }}
 {\mbox{\ with \ }}  (i,j) {\mbox{\   WB
}}\}. \nonumber
\end{eqnarray}
To introduce the Feng-Rao bound for primary codes we define for $i=1,2,\ldots , n$
\begin{eqnarray}
\widebar{\sigma}^{\textrm{WB}}_{({\mathcal{B}},{\mathcal{U}})}(i)&=&\sharp\{ l \in \{1, 2, \ldots , n\} \mid \widebar{\rho}(\bm{b}_{i}\ast \bm{u}_j) = l \ {\mbox{for
    some}} \   \bm{u}_j \in {\mathcal{U}}\nonumber \\
&& {\mbox{\ \hspace{6.5cm} with \ }} \
(i,j)  {\mbox{\  WB}}\}, \nonumber \\
\widebar{\sigma}^{\textrm{OWB}}_{({\mathcal{B}},{\mathcal{U}})}(i)&=&\sharp \{ l \in \{1,2,\ldots , n\} \mid \widebar{\rho}(\bm{b}_{i}\ast \bm{u}_j) = l \ {\mbox{for
    some}} \   \bm{u}_j \in {\mathcal{U}}\nonumber \\
&& {\mbox{\ \hspace{6.5cm} with \ }} \
(i,j)  {\mbox{\  OWB}}\}.\nonumber
\end{eqnarray}

In the next theorem the first bound is the Feng-Rao bound for dual codes. The latter is 
the Feng-Rao bound for primary codes.
\begin{theorem}\label{thebounds}
The minimum distance of $C^\perp({\mathcal{B}},I)$ satisfies:
\begin{eqnarray}
d(C^\perp({\mathcal{B}},I))&\geq&\min \{\widebar{\mu}^{\textrm{OWB}}_{({\mathcal{B}},{\mathcal{U}})}(l) \mid l \notin
I\}\label{bounddual1}\\
&\geq&\min \{\widebar{\mu}^{\textrm{WB}}_{({\mathcal{B}},{\mathcal{U}})}(l) \mid l \notin
I\}.\label{bounddual3}
\end{eqnarray}
The minimum distance of $C({\mathcal{B}},I)$ satisfies:
\begin{eqnarray}
d(C({\mathcal{B}},I))&\geq&\min \{\widebar{\sigma}^{\textrm{OWB}}_{({\mathcal{B}},{\mathcal{U}})}(i) \mid i \in I\} \label{boundprimary1}\\
&\geq&\min \{\widebar{\sigma}^{\textrm{WB}}_{({\mathcal{B}},{\mathcal{U}})}(i) \mid i \in I\}. \label{boundprimary3}
\end{eqnarray}
\end{theorem}
\begin{proof}
For proofs of the bounds~(\ref{bounddual1}) and (\ref{boundprimary1})
see~\cite[Th.\ 1]{geithom}. 
\qed
\end{proof}

The bound~(\ref{boundprimary1}) is a slight enhancement 
of
(\ref{boundprimary3}), the latter being introduced for the first time
in~\cite{AG}. This explains why we in Section~\ref{sec1} 
 referred to both bounds as Andersen and Geil's
bound. The bound~(\ref{bounddual3}) is a special case of
the three-bases formulation
\cite[Sec.\ 4.3]{handbook}, \cite[Sec.\ 2]{MM} of the original Feng-Rao
bound~\cite{FR24,FR1,FR2}. The formulation in~\cite{handbook,MM} involves
three
bases rather than only the above two.
We note that
three bases were also used earlier in \cite{early}
for expressing the idea of Feng and Rao \cite{FR24} in the
general linear code formulation. 
The contribution of~\cite{geithom} is the
notion of OWB as a generalization of WWB \cite{miura1,miura2}. If we replace OWB
with WWB in~(\ref{bounddual1}) then we get another special case of \cite{MM}.

Note, that Theorem~\ref{thebounds}
allows us to construct improved codes by choosing $I$ cleverly
according to the $\widebar{\mu}$ respectively $\widebar{\sigma}$
values. Remark~\ref{remrom} demonstrates that Theorem~\ref{thebounds}
can also be formulated in a version with WWB instead of OWB or
WB (See~\cite{MM}). However, the result to be shown in the next section that
(\ref{bounddual1}) and (\ref{boundprimary1}) 
respectively (\ref{bounddual3}) and (\ref{boundprimary3}) are
consequences of each other seemingly does not hold for WWB.

\section{The two bounds are consequences of each other}\label{secthree}
In this section we consider two bases ${\mathcal{G}}=\{\bm{g}_1, \ldots ,
\bm{g}_n\}$ and ${\mathcal{H}}=\{\bm{h}_1, \ldots , \bm{h}_n\}$ for
${\mathbb{F}}_q^n$. They shall both be used in replacement of the
${\mathcal{B}}$ from the previous section. The basis ${\mathcal{U}}=\{ \bm{u}_1,\ldots
\bm{u}_n\}$ 
will be fixed throughout the section. Hence, we will be concerned with
WB and OWB pairs with respect to $({\mathcal{G}},{\mathcal{U}})$ as
well as with respect to $({\mathcal{H}},{\mathcal{U}})$.\\
We shall assume the following strong relation
\begin{equation}
\bm{g}_i\cdot\bm{h}_j=\delta_{i,n-j+1} \label{dualbases}
\end{equation}
where the expression on the left side is the inner product and the
expression on the right side is Kronecker's delta. Clearly, if ${\mathcal{G}}$ is
given then there is a unique choice of ${\mathcal{H}}$ such that~(\ref{dualbases})
holds and vice versa. Let $G$ be an $n \times n$ matrix where
row $i$ equals $\bm{g}_i$ and let $H$ be an $n \times n$ matrix where
column $j$ equals $\bm{h}_{n-j+1}^T$. Then indeed $H=G^{-1}$. 

The above correspondence gives us an alternative expression for the
function $\widebar{\rho}$ (Definition~\ref{defone}), namely that for any
non-zero $\bm{v}$ 
$$\widebar{\rho}_{\mathcal{G}}(\bm{v})=\max \{k \mid \bm{v} \cdot \bm{h}_{n-k+1} \neq
0\},$$
$$\widebar{\rho}_{\mathcal{H}}(\bm{v})=\max \{k \mid \bm{v} \cdot \bm{g}_{n-k+1} \neq
0\}.$$
For $I \subseteq \{1, \ldots , n\}$ denote 
$$\widebar{I}=\{1, \ldots , n\}
\backslash \{n-i+1 \mid i \in I\}$$ 
and observe
$\widebar{\widebar{I}}=I$. Equation~(\ref{dualbases}) corresponds to saying
\begin{eqnarray}
C({\mathcal{G}},I)&=&C^\perp({\mathcal{H}},\widebar{I})\nonumber 
\end{eqnarray}
which can of course also be formulated
\begin{eqnarray}
C({\mathcal{G}},\widebar{I})&=&C^\perp({\mathcal{H}},I).\nonumber 
\end{eqnarray}
We shall show that for $i=1, \ldots , n$ it holds that
\begin{eqnarray}
\widebar{\mu}^{\textrm{WB}}_{({\mathcal{H}},{\mathcal{U}})}(n-i+1)&=&\widebar{\sigma}^{\textrm{WB}}_{({\mathcal{G}},{\mathcal{U}})}(i),
\nonumber \\ 
\widebar{\mu}^{\textrm{OWB}}_{({\mathcal{H}},{\mathcal{U}})}(n-i+1)&=&\widebar{\sigma}^{\textrm{OWB}}_{({\mathcal{G}},{\mathcal{U}})}(i)\nonumber 
\end{eqnarray}
which is to say that for WB and OWB the
Feng-Rao bound for dual codes and the Feng-Rao bound for primary codes are consequences of each
other. The fact that a similar result seemingly does not hold for WWB might partially explain why the
correspondences  of the present paper has not been found before. 

\begin{lemma}\label{lemlom1}
The following statements are equivalent:
\begin{enumerate}
\item $\widebar{\rho}_{\mathcal{G}}(\bm{g}_i \ast \bm{u}_j)=k$\\ and $(i,j)$
  is WB with respect to $({\mathcal{G}},{\mathcal{U}})$.
\item $\widebar{\rho}_{\mathcal{H}}(\bm{h}_{n-k+1} \ast \bm{u}_j)=n-i+1$\\ and $(n-k+1,j)$
  is WB with respect to $({\mathcal{H}},{\mathcal{U}})$.
\end{enumerate}
\end{lemma}
\begin{proof}
We shall make extensive use of the correspondence
\begin{equation}
(\bm{v} \ast \bm{w}) \cdot \bm{s}=(\bm{s}\ast \bm{w})\cdot \bm{v}. \label{cirkus}
\end{equation}
Assume that statement 1 holds. That is
assume $\widebar{\rho}_{\mathcal{G}}(\bm{b}_i \ast \bm{u}_j)=k$ and
$(i,j)$ is WB with respect to $({\mathcal{G}},{\mathcal{U}})$. From
the first half of this assumption we get
\begin{eqnarray}
(\bm{g}_i \ast \bm{u}_j) \cdot \bm{h}_{n-k+1} &\neq 0, \label{trek1}\\
(\bm{g}_i \ast \bm{u}_j) \cdot \bm{h}_{t}&=0&{\mbox{ \ for \ }} t <n-k+1, \label{trek2}
\end{eqnarray}
and from the latter half
\begin{eqnarray}
(\bm{g}_{i^\prime} \ast \bm{u}_j) \cdot \bm{h}_{t}&=0& {\mbox{ \ for \ }}
i^\prime <i {\mbox{ \ and \ }}t \leq n-k+1,\label{trek3}\\
(\bm{g}_i \ast \bm{u}_{j^\prime}) \cdot \bm{h}_{t}&=0& {\mbox{ \ for \
  }} j^\prime <j {\mbox{ \ and \ }}t \leq n-k+1,\nonumber \\
(\bm{g}_{i^\prime} \ast \bm{u}_{j^\prime}) \cdot \bm{h}_{t}&=0& {\mbox{ \ for \
  }} i^\prime < i,  j^\prime <j {\mbox{ \ and \ }}t \leq
n-k+1.\nonumber \\ 
\end{eqnarray}
Using
(\ref{trek1}), (\ref{trek3}) in combination with
(\ref{cirkus}) gives
\begin{eqnarray}
(\bm{h}_{n-k+1}\ast \bm{u}_j)\cdot \bm{g}_i &\neq& 0,\label{star1}\\
(\bm{h}_{n-k+1}\ast \bm{u}_j)\cdot \bm{g}_{i^\prime}& =&0 {\mbox{ \ for
    \ }} i^\prime < i .\label{star2}
\end{eqnarray}
In a similar fashion we derive at
\begin{eqnarray}
(\bm{h}_{t}\ast \bm{u}_j)\cdot \bm{g}_{i^\prime} &=0& {\mbox{ \ for \
  }} t < n-k+1 {\mbox{ \ and \ }} i^\prime \leq i,\label{star3} \\
(\bm{h}_{n-k+1}\ast \bm{u}_{j^\prime})\cdot \bm{g}_{i^\prime}& =0& {\mbox{ \ for \
  }} j^\prime < j {\mbox{ \ and \ }} i^\prime \leq i,\label{star4} \\
(\bm{h}_{t}\ast \bm{u}_{j^\prime})\cdot \bm{g}_{i^\prime} &=0& {\mbox{ \ for \
  }} t < n-k+1, j^\prime < j, {\mbox{ \ and \ }} i^\prime \leq i.\label{star5}
\end{eqnarray}
Expressions (\ref{star1}) and (\ref{star2}) mean that
$\widebar{\rho}_{\mathcal{H}}(\widebar{h}_{n-k+1}\ast \bm{u}_j)=n-i+1$ and
(\ref{star3}), (\ref{star4}), (\ref{star5}) imply that $(n-k+1,j)$ is WB
with respect to $({\mathcal{H}},{\mathcal{U}})$. In other words
statement 2 is true. That statement 2 implies statement 1 follows by symmetry.
\qed
\end{proof}

\begin{lemma}\label{lemlom2}
The following statements are equivalent
\begin{enumerate}
\item $\widebar{\rho}_{\mathcal{G}}(\bm{g}_i \ast \bm{u}_j)=k$\\ and $(i,j)$
  is OWB with respect to $({\mathcal{G}},{\mathcal{U}})$.
\item $\widebar{\rho}_{\mathcal{H}}(\bm{h}_{n-k+1} \ast \bm{u}_j)=n-i+1$\\ and $(n-k+1,j)$
  is OWB with respect to $({\mathcal{H}},{\mathcal{U}})$.
\end{enumerate}
\end{lemma}
\begin{proof}
Assume that statement 1 holds. The first part implies~(\ref{trek1})
and (\ref{trek2}) and from the latter part we get
(\ref{trek3}). Using (\ref{trek1}), (\ref{trek3}) in combination
with (\ref{cirkus}) gives (\ref{star1}) and (\ref{star2}). Combining
(\ref{trek2}), (\ref{trek3}) with (\ref{cirkus}) we get
(\ref{star3}). 
Expressions (\ref{star1}) and (\ref{star2}) mean that
$\widebar{\rho}_{\mathcal{H}}(\widebar{h}_{n-k+1}\ast \bm{u}_j)=n-i+1$ and
(\ref{star3}) implies that $(n-k+1,j)$ is OWB
with respect to $({\mathcal{H}},{\mathcal{U}})$. In other words
statement 2 is true. That statement 2 implies statement 1 follows by symmetry.
\qed
\end{proof}

\begin{theorem}\label{kjlj}
Assume that ${\mathcal{G}},{\mathcal{H}}$ satisfy
condition~(\ref{dualbases}). Let a non-empty set $I \subseteq \{1, 2,
\ldots , n\}$ be given. Then 
$C({\mathcal{G}},I)=C^\perp({\mathcal{H}},\widebar{I})$ and
\begin{eqnarray}
\min \{ \widebar{\mu}_{({\mathcal{H}},{\mathcal{U}})}^{\textrm{OWB}}(l) \mid l \neq
\widebar{I} \} &=&\min\{\widebar{\sigma}_{({\mathcal{G}},{\mathcal{U}})}^{\textrm{OWB}}(i)
\mid i \in I\}, \label{eq78} \\
\min \{ \widebar{\mu}_{({\mathcal{H}},{\mathcal{U}})}^{\textrm{WB}}(l) \mid l \neq
\widebar{I} \} &=&\min\{\widebar{\sigma}_{({\mathcal{G}},{\mathcal{U}})}^{\textrm{WB}}(i)
\mid i \in I\}. \label{eq79} 
\end{eqnarray}
Similar results hold with the role of $I$ and $\widebar{I}$ interchanged.
\end{theorem}
\begin{proof}
Follows from Lemma~\ref{lemlom1}, Lemma~\ref{lemlom2} and the fact
that $i\in I \Leftrightarrow n-i+1 \notin \widebar{I}$.
\qed
\end{proof} 

\begin{remark}\label{rem77}
Consider the following assumption which is weaker than~(\ref{dualbases}):
\begin{equation}
\left\{ 
\begin{array}{ll}
\bm{g}_i \cdot \bm{h}_{n-i+1} \neq 0, \\
\bm{g}_i \cdot \bm{h}_j=0& {\mbox{ \ \ for \ }} j <n-i+1.
\end{array} 
\right.
\label{eq77}
\end{equation}
Inspecting the proofs of Lemma~\ref{lemlom1} and \ref{lemlom2} we see
that~(\ref{eq78}) and (\ref{eq79}) still hold. However,
$C({\mathcal{G}},I)=C^\perp({\mathcal{H}},\widebar{I})$ is only 
guaranteed to hold when $I$ is of the form $I=\{1, 2, \ldots , k\}$
under the assumption (\ref{eq77}) that is
weaker than (\ref{dualbases}). Therefore (\ref{eq77}) does not allow us
to translate information regarding improved primary codes $C({\mathcal{G}},I)$
and improved dual codes $C^\perp({\mathcal{H}},\widebar{I})$.
\end{remark}

The above remark is in the spirit of \cite{duursma94}
where Duursma showed how to speed up the erasure decoding of
algebraic geometric codes,
by introducing a condition \cite[Def.\ 4]{duursma94} that is
implied by (\ref{dualbases}) and implies (\ref{eq77}).
After having finished the manuscript we learned that Duursma is aware of
the implication (\ref{eq79}) of (\ref{eq77}) to the
Feng-Rao bound for the case of non-improved codes, i.e.\ 
$I$ is of the form $\{1$, \ldots, $k\}$.

The following example illustrates the last part of the remark.
\begin{example}
Let $n=4$ and $I=\{3,4\}$. This gives $\widebar{I}=\{3,4\}$. Assume ${\mathcal{G}}$
and ${\mathcal{H}}$ are given such that (\ref{eq77}) is satisfied. We keep ${\mathcal{H}}$, leaving
$C^\perp ({\mathcal{H}},\widebar{I})$ intact, but allow for redefinition of
${\mathcal{G}}$ by replacing for $i=1, \ldots , 4$, $\bm{g}_i$ with any other
vector of $\widebar{\rho}_{\mathcal{G}}$-value equal to $i$. This possibly redefines
$C({\mathcal{G}},I)$ but keeps the estimates of the minimum distance
the same. Even allowing for the above redefinition (\ref{eq77}) is not
enough to guarantee that for any of the above choices of $g_3,
g_4$ we arrive at $\bm{h}_i \cdot \bm{g_j}=0$ for all $i=3,4$ and
$j=3,4$. Hence, we have no guarantee that all of the resulting codes
$C({\mathcal{G}},I)$ equals $C^\perp ({\mathcal{H}},\widebar{I})$.
\end{example}

\section{Utilizing a supporting algebra}\label{secfour}
As mentioned in the introduction the results in Section~\ref{sectwo}
and \ref{secthree} are universal in the sense that they can in particular be
applied when given various types of supporting algebra. We now explain how to do
this in the case of the supporting algebra being an order
domain. We shall restrict to order functions that are weight
functions. From~\cite[Sec.\ 6]{AG} we have the following couple of results.

\begin{definition}\label{defweightfunction}
Let $R$ be
an ${\mathbb{F}}_q$-algebra and let $\Gamma$ be a subsemigroup of
${\mathbb{N}}_{0}^r$ for some $r$. Let $\prec$ be a monomial ordering
on ${\mathbb{N}}_0^r$. A surjective map $\rho : R \rightarrow \Gamma_{-
  \infty}=\Gamma \cup \{- \infty\}$ that satisfies the following six
conditions is said to be a weight function
$$
\begin{array}{rl}
(W.0)&\rho(f)=-\infty \ {\mbox{if and only if}} \ f=0.\\
(W.1)&\rho(af)=\rho(f) \ {\mbox{for all non-zero}} \ a \in
{\mathbb{F}}_q.\\
(W.2)&\rho(f+g) \preceq \max \{ \rho(f), \rho(g) \}\ {\mbox{ and equality
holds when}} \ \rho(f) \prec \rho(g).\\
(W.3)& {\mbox{If}} \  \rho(f) \prec \rho(g) \ {\mbox{and}} \ h \neq 0,
\ {\mbox{then}} \ \rho(fh) \prec \rho(gh).\\
(W.4)&{\mbox{If $f$ and $g$ are non-zero and $\rho(f)=\rho(g)$, then
    there}}\\
&{\mbox{exists a non-zero $a \in {\mathbb{F}}_q$ such that $\rho(f-ag)
    \prec \rho(g)$}}.\\
(W.5)&{\mbox{If $f$ and $g$ are non-zero then \ }} \rho(fg)=\rho(f)+\rho(g).
\end{array}
$$
An ${\mathbb{F}}_q$-algebra with a weight function is called an order
domain over ${\mathbb{F}}_q$.  
\end{definition}

\begin{theorem}\label{above}
Given a weight function then any
set
${\mathcal{B}}=\{f_{\gamma} \mid \rho(f_{\gamma})=\gamma \}_{\gamma \in
\Gamma}$ constitutes a basis for $R$ as a vector space over
${\mathbb{F}}_q$. 
In particular $\{f_{\lambda} \in {\mathcal{B}} \mid
\lambda \preceq \gamma \}$ constitutes a basis for $R_{\gamma}=\{f
\in R
\mid \rho(f) \preceq \gamma \}$.
\qed
\end{theorem}

\begin{definition}\label{defmorphism}
Let $R$ be an ${\mathbb{F}}_q$-algebra. A surjective map $\varphi : R \rightarrow
\vr{q}{n}$ is called a morphism of ${\mathbb{F}}_q$-algebras if $\varphi$ is
${\mathbb{F}}_q$-linear and $\varphi(fg)=\varphi(f) \ast \varphi(g)$ for all $f,
g \in R$.
\end{definition}

\begin{definition}\label{defalpha}\label{ovenfor}
Assume that $\varphi$ is a morphism as
in Definition~\ref{defmorphism}. Let 
$\alpha(1)={\bm{0}}$. For $i=2, 3, \ldots ,n$ define recursively $\alpha(i)$
to be the smallest element in $\Gamma$ that is greater than $\alpha(1),
\alpha(2), \ldots ,\alpha(i-1)$ and satisfies $\varphi(R_{\gamma})
\subsetneq \varphi (R_{\alpha(i)})$ for all $\gamma \prec
\alpha(i)$. Write $\Delta(R,\rho,\varphi)=\{\alpha(1), \alpha(2),
\ldots ,\alpha(n)\}$.
\end{definition}

The following theorem is a synthesis \cite[Th.~25, Pro.~27, 28]{AG}.

\begin{theorem}\label{thesnabel}
Let $\Delta(R,\rho,\varphi)=\{\alpha(1), \alpha(2),
\ldots ,\alpha(n)\}$ be as in Definition~\ref{ovenfor}. The set
\begin{equation}
{\mathcal{B}}=\{ \bm{b}_1=\varphi (f_{\alpha(1)}), \bm{b}_2=\varphi
(f_{\alpha(2)}), \ldots, \bm{b}_n=\varphi (f_{\alpha(n)})\}
\label{snabel}
\end{equation} 
constitutes a basis for ${\mathbb{F}}_{q}^{n}$ as a vector space over
${\mathbb{F}}_{q}$.

If $\gamma , \lambda \in \Gamma$ 
satisfy $\gamma + \lambda =\alpha(l)$ for some $\alpha(l) \in
\Delta(R,\rho,\varphi)$ 
then
$\gamma=\alpha(i)$ and $\lambda=\alpha(j)$ for some $i,j \in \{1,
\ldots , n\}$ and $(i,j)$ is WB with respect to $({\mathcal{B}},{\mathcal{B}})$.
\qed
\end{theorem}
Theorem~\ref{thesnabel} in combination with
  (\ref{boundprimary3}) respectively (\ref{bounddual3}) proves the
  order bound for primary \cite{AG} respectively dual \cite{handbook}
  order domain codes. To formulate these bounds we will need a definition.

\begin{definition}\label{defmusique}
Let notation be as in Definition~\ref{ovenfor}. For $\lambda \in
\Delta(R, \rho ,\varphi )$ we define
$$\sigma (\lambda) = \sharp \{ \eta \in \Delta (R, \rho , \varphi )
\mid \eta = \lambda + \gamma {\mbox{ \ for some \ }} \gamma \in
\Gamma\}.$$
For $\eta \in \Delta(R, \rho , \varphi )$ define
$$\mu (\eta ) = \sharp \{ \lambda \in \Gamma \mid \lambda +\gamma =
\eta {\mbox{ \ for some \ }}\gamma \in \Gamma \}.$$ 
\end{definition}

The bounds are:
\begin{theorem}\label{theorderbounds}
Let ${\mathcal{B}}$ be as in Theorem~\ref{thesnabel}. The minimum
distance of $C({\mathcal{B}},I)$ is at least $\min \sigma (\alpha (i))
\mid i \in I\}$. The minimum distance of $C^\perp ({\mathcal{B}},I)$
is at least $\min \{ \mu (\alpha(l)) \mid l \in \{ 1, \ldots , n\}
\backslash I\}$. \qed
\end{theorem}

Just as with the Feng-Rao bounds, the order bounds suggest improved
code constructions by choosing $I$ cleverly.
The best-known example of an order domain is the
following. Consider an algebraic function field over ${\mathbb{F}}_q$
of transcendence degree $1$. Let $P_1, \ldots , P_n,Q$ be rational
places. Then $R=\sum_{m=0}^\infty {\mathcal{L}}(mQ)$ is an order
domain with a weight function $\rho(f)=-\nu_Q(f)$. Here, $\nu_Q$ is
the discrete valuation corresponding to $Q$. Let $\varphi$ be defined
by $\varphi(f)=(f(P_1), \ldots , f(P_n))$. Applying the standard
notation for Weierstrass semigroups we have $\Gamma=H(Q)$. The
corresponding set
$\Delta(R,\rho,\varphi)$ sometimes in the literature is denoted
$H^\ast(Q)$ \cite{GMRT}. 

It is well-known \cite{petertom,MM} how to decode primary codes defined from algebraic function
fields of transcendence degree $1$. In Section~\ref{secfive} we shall
see how to decode primary codes coming form 
order
domains of higher transcendence degree. Examples of order domains of
higher transcendence degree can be found in
\cite{AG,galindo1,galindo2,galindo3,GP,little,OSullivan}. 
The most simple example
of such an order domain is $R={\mathbb{F}}_q[X_1, \ldots , X_m]$ with
a weight function given by $\rho(X_1^{i_1} \cdots X_m^{i_m})=(i_1,
\ldots , i_m)$ and the ordering $\prec$
(Definition~\ref{defweightfunction}) being any fixed monomial
ordering. Evaluating in all the points of ${\mathbb{F}}_q^{m}$ we get
$q$-ary Reed-Muller codes, Massey-Costello-Justesen codes (also known
as hyperbolic codes) and weighted Reed-Muller codes. Theoretical
computer scientists take a special interest
in the case where rather
than evaluating in all of ${\mathbb{F}}_q^{m}$ one evaluates in an
arbitrary Cartesian product $S_1 \times \cdots \times S_m$, where $S_1,
\ldots ,S_m \subseteq {\mathbb{F}}_q$. Such codes can have
surprisingly good parameters as demonstrated
in~\cite{weightedreedmuller}.
Their good parameters cannot be
obtained as the punctured codes of the ordinary
Reed-Muller codes. The following two examples illustrate
the theory from Section~\ref{secthree} when $R={\mathbb{F}}_q[X,Y]$
and the Cartesian product $S_1 \times S_2$ is a proper subset of ${\mathbb{F}}_q^{2}$.

\begin{example}\label{ex2}
Consider $R={\mathbb{F}}_5[X,Y]$, $\rho(X^iY^j)=(i,j)$ and let $\prec$
be the graded lexicographic ordering with $(1,0)\prec (0,1)$. Consider the
point ensemble
\begin{eqnarray}
\{P_1=(1,1),P_2=(1,2),P_3=(1,3),P_4=(2,1), \ldots ,P_9=(3,3)\} {\mbox{ \ \hspace{1cm}}}
\nonumber \\
 =\{1,2,3\} \times \{1,2,3\} \subseteq {\mathbb{F}}_5^2\nonumber
\end{eqnarray}
and the morphism ${\textrm{ev}}: {\mathbb{F}}_5[X.Y]
\rightarrow {\mathbb{F}}_5^9$, given by ${\textrm{ev}}(F)=(F(P_1), \ldots ,
F(P_9))$. The basis~(\ref{snabel}) becomes (writing ${\mathcal{G}}$
instead of ${\mathcal{B}}$ as we want to apply the theory from Section~\ref{secthree})
\begin{eqnarray}
{\mathcal{G}}&=&\{\bm{g}_1={\textrm{ev}}(1),\bm{g}_2={\textrm{ev}}(X),\bm{g}_3={\textrm{ev}}(Y),\bm{g}_4={\textrm{ev}}(X^2),\bm{g}_5={\textrm{ev}}(XY),\nonumber
\\
&& {\mbox{ \ \hspace{1cm}}} \bm{g}_6={\textrm{ev}}(Y^2),\bm{g}_7={\textrm{ev}}(X^2Y),\bm{g}_8={\textrm{ev}}(XY^2),\bm{g}_9={\textrm{ev}}(X^2Y^2)\}. \nonumber
\end{eqnarray}
Still using the notation from Section~\ref{secthree} let
${\mathcal{U}}={\mathcal{G}}$. From Theorem~\ref{thesnabel} we
immediately get
information about WB pairs as described in the following array. Here a
number $x$ in entry $(i,j)$ means that $(i,j)$ is known to be WB and
that $\widebar{\rho}_{\mathcal{G}}(\bm{g}_i \ast \bm{g}_j)=x$.
\begin{equation}
\left[
\begin{array}{rrrrrrrrr}
1&\, \, 2&\, \, 3&\, \, 4&\,\,  5&\,\,  6&\,\,  7&\,\,  8&\,\,  9\\
2&4&5&&7&8&&9\\
3&5&6&7&8&&9\\
4&&7&&&9\\
5&7&8&&9\\
6&8&&9\\
7&&9\\
8&9\\
9
\end{array}
\right]
\label{eqarray}
\end{equation}
It is now an easy task to calculate from this array
\begin{eqnarray}
\begin{array}{lll}
\widebar{\sigma}_{({\mathcal{G}},{\mathcal{G}})}^{{\textrm{WB}}}(1)=9, {\mbox{ \ }}&\widebar{\sigma}_{({\mathcal{G}},{\mathcal{G}})}^{{\textrm{WB}}}(2)\geq 6, {\mbox{ \ }}&\widebar{\sigma}_{({\mathcal{G}},{\mathcal{G}})}^{{\textrm{WB}}}(3)\geq 6,\\
\widebar{\sigma}_{({\mathcal{G}},{\mathcal{G}})}^{{\textrm{WB}}}(4)\geq 3, {\mbox{ \ }}&\widebar{\sigma}_{({\mathcal{G}},{\mathcal{G}})}^{{\textrm{WB}}}(5)\geq 4, {\mbox{ \ }}&\widebar{\sigma}_{({\mathcal{G}},{\mathcal{G}})}^{{\textrm{WB}}}(6)\geq 3,\\
\widebar{\sigma}_{({\mathcal{G}},{\mathcal{G}})}^{{\textrm{WB}}}(7)\geq 2, {\mbox{ \ }}&\widebar{\sigma}_{({\mathcal{G}},{\mathcal{G}})}^{{\textrm{WB}}}(8)\geq 2, {\mbox{ \ }}&\widebar{\sigma}_{({\mathcal{G}},{\mathcal{G}})}^{{\textrm{WB}}}(9)\geq 1,\\
\end{array}
\nonumber 
\end{eqnarray}
(actually, equality holds -- also even if we replace WB with OWB -- but
we will not prove this fact). Using a computer  we found
$\bm{h}_1, \ldots , \bm{h}_9$ and the corresponding polynomials. We
have:
\begin{eqnarray}
\bm{h}_1&=&{\textrm{ev}}(X^2Y^2+XY^2+X^2Y+XY), \nonumber \\
\bm{h}_2&=&{\textrm{ev}}(X^2Y^2+3XY^2+X^2Y+Y^2+3XY+Y), \nonumber \\
\bm{h}_3&=&{\textrm{ev}}(X^2Y^2+XY^2+3X^2Y+3XY+X^2+X), \nonumber \\
\bm{h}_4&=&{\textrm{ev}}(XY^2+Y^2+XY+Y), \nonumber \\
\bm{h}_5&=&{\textrm{ev}}(X^2Y^2+3XY^2+3X^2Y+Y^2+4XY+X^2+3Y+3X+1), \nonumber \\
\bm{h}_6&=&{\textrm{ev}}(X^2Y+XY+X^2+X), \nonumber \\
\bm{h}_7&=&{\textrm{ev}}(XY^2+Y^2+3XY+3Y+X+1), \nonumber \\
\bm{h}_8&=&{\textrm{ev}}(X^2Y+3XY+X^2+Y+3X+1), \nonumber \\
\bm{h}_9&=&{\textrm{ev}}(XY+Y+X+1). \nonumber 
\end{eqnarray}
It is clearly not at all an easy task to determine by hand  WB pairs with
respect to $({\mathcal{H}},{\mathcal{G}})$
without using (\ref{eqarray}) and the correspondence in
Lemma~\ref{lemlom1}. A similar observation holds regarding  
$\widebar{\rho}_{\mathcal{H}}(\bm{h}_i \ast \bm{g}_j)$. 
\end{example}
\begin{example}\label{ex3}
Consider $R={\mathbb{F}}_4[X,Y]$ with $\rho$ and $\prec$ as in
Example~\ref{ex2}. Write ${\mathbb{F}}_4={\mathbb{F}}_2[T]\backslash \langle
T^2+T+1\rangle$ and let $\alpha=T+\langle T^2+T+1\rangle$. The point
ensemble
\begin{eqnarray}
\{P_1=(0,1), P_2=(0,\alpha),P_3=(1,1), P_4=(1,\alpha),P_5=(\alpha,1),P_6=(\alpha
,\alpha ) \} \nonumber \\
\,\,\, =\{0,1,\alpha\} \times \{1,\alpha\} \subseteq {\mathbb{F}}_4^2 \nonumber
\end{eqnarray}
defines a morphism ${\textrm{ev}}: {\mathbb{F}}_4[X,Y]
\rightarrow {\mathbb{F}}_4^6$ by ${\textrm{ev}}(F)=(F(P_1), \ldots ,
F(P_6))$. The basis~(\ref{snabel}) becomes
\begin{eqnarray}
{\mathcal{G}}&=&\{
\bm{g}_1={\textrm{ev}}(1),\bm{g}_2={\textrm{ev}}(X),\bm{g}_3={\textrm{ev}}(Y),\nonumber \\
&&\bm{g}_4={\textrm{ev}}(X^2),\bm{g}_5={\textrm{ev}}(XY),\bm{g}_6={\textrm{ev}}(X^2Y)\}. \nonumber
\end{eqnarray}
Again we let ${\mathcal{U}}={\mathcal{G}}$. Using
Theorem~\ref{thesnabel} we easily derive information on
$\widebar{\rho}_{\mathcal{G}}$ as well as information on which pairs are
WB. From this we conclude
$$
\begin{array}{lll}
\widebar{\sigma}_{({\mathcal{G}},{\mathcal{G}})}^{{\textrm{WB}}}(1)=6, {\mbox{ \ }} &\widebar{\sigma}_{({\mathcal{G}},{\mathcal{G}})}^{{\textrm{WB}}}(2)\geq
4,  {\mbox{ \ }}&\widebar{\sigma}_{({\mathcal{G}},{\mathcal{G}})}^{{\textrm{WB}}}(3)\geq 3, \\
\widebar{\sigma}_{({\mathcal{G}},{\mathcal{G}})}^{{\textrm{WB}}}(4)\geq
2,  {\mbox{ \ }}&\widebar{\sigma}_{({\mathcal{G}},{\mathcal{G}})}^{{\textrm{WB}}}(5)\geq 2,  {\mbox{ \ }} &\widebar{\sigma}_{({\mathcal{G}},{\mathcal{G}})}^{{\textrm{WB}}}(6)=1,
\end{array}
$$
(again actually equalities hold). Using a computer we found
$\bm{h}_1, \ldots , \bm{h}_6$ and the polynomials behind them. They
are:
\begin{eqnarray}
\bm{h}_1&=&{\textrm{ev}}(\alpha X +1  ), \nonumber \\
\bm{h}_2&=&{\textrm{ev}}(\alpha X^2+\alpha+1  ), \nonumber \\
\bm{h}_3&=&{\textrm{ev}}(\alpha XY+Y+X+\alpha +1  ), \nonumber \\
\bm{h}_4&=&{\textrm{ev}}(X^2 +(\alpha +1)X+\alpha ), \nonumber \\
\bm{h}_5&=&{\textrm{ev}}(\alpha X^2Y+X^2+(\alpha +1)Y+\alpha  ), \nonumber \\
\bm{h}_6&=&{\textrm{ev}}( X^2Y+(\alpha +1)XY+(\alpha +1)X^2+\alpha
Y+\alpha X +1 ). \nonumber 
\end{eqnarray}
Some information on the values of $\widebar{\mu}_{({\mathcal{H}},{\mathcal{G}})}^{{\textrm{WB}}}(1),
\ldots , \widebar{\mu}_{({\mathcal{H}},{\mathcal{G}})}^{{\textrm{WB}}}(6)$ can be seen directly from
${\mathcal{H}}$ and ${\mathcal{G}}$, but one does not get the complete picture without
using  Lemma~\ref{lemlom2} or alternatively a computer.
\end{example}

\section{Feng-Rao decoding}\label{secfive}
We now turn to the problem of decoding primary codes up
to half the Feng-Rao bound. In this section we assume
that two bases ${\mathcal{G}},{\mathcal{U}}$ are given. We consider a primary
  code $C({\mathcal{G}},I)$ and assume that we have information about
  WB pairs with respect to $({\mathcal{G}},{\mathcal{U}})$. To
  set up the decoding algorithm, the first task is to calculate the
  basis ${\mathcal{H}}$ such that the correspondence~(\ref{dualbases})
  holds. Using Lemma~\ref{lemlom1} we then translates the information we
  have on WB pairs with respect to
  $({\mathcal{G}},{\mathcal{U}})$ into information on WB
  pairs with respect to $({\mathcal{H}},{\mathcal{U}})$. Decoding  $C({\mathcal{G}},I)$ up to half the Feng-Rao bound
  for primary codes by Theorem~\ref{kjlj} now is the same as decoding
  $C^\perp({\mathcal{H}},\widebar{I})$ up to half the Feng-Rao bound for
  dual codes. The contribution of the present paper regarding decoding
  basically is this observation. 
In~\cite{handbook} H{\o}holdt, van Lint and Pellikaan presented 
a couple of decoding algorithms that correct up to
half the order bound for dual codes. Of particular interest to us is the algorithm
in~\cite[Sec.\ 6.3]{handbook}. This algorithm uses majority voting for
unknown syndromes and builds on the work by Feng and Rao~\cite{FR24} and
Duursma \cite{D18,D19}. For this algorithm to be applicable to the
situation of general codes $C^\perp({\mathcal{H}},\widebar{I})$ a couple of
modifications are needed. These modifications were described by
Matsumoto and Miura in~\cite[Sec.\ 2]{MM}. Actually, the description
in~\cite{MM} is a little more general than we need as three bases are
involved in their description in contrast to our two
${\mathcal{H}},{\mathcal{U}}$. Obviously, the problem is solved by
letting two of them being the same. By~\cite[Prop.\ 2.5]{MM} the
algorithm can correct up to
$$(\min\{\widebar{\mu}_{({\mathcal{H}},{\mathcal{U}})}^{\textrm{WB}}(l) \mid l \notin
\widebar{I}\}-1)/2$$ errors in computational complexity ${\mathcal{O}}(n^3)$. 

As the description in~\cite[Sec.\ 2]{MM} is rather brief in the following we shall give an overview of how to apply the
above algorithm to our problem and
illustrate it with an example (we refer to~\cite{handbook,MM} for
more details). Observe that in practice we might only have partial information
on which pairs $(i,j)$ are WB. This is for instance the case when our
information comes from the study of a supporting algebra. Clearly we
can derive exact information on all WB pairs $(i,j)$ and corresponding
values $\widebar{\rho}_{\mathcal{H}}(\bm{h}_i \ast \bm{u}_j)$ in a
preparation step. This can be done in ${\mathcal{O}}(n^4)$ operations
which is more expensive than the algorithm itself. To treat also the
case where one wants to avoid such a preparation step we define
\begin{eqnarray}
\tilde{\sigma}^{\textrm{WB}}_{({\mathcal{G}},{\mathcal{U}})}(i)&=&\sharp\{ l \in \{1, 2, \ldots , n\} \mid \widebar{\rho}(\bm{g}_{i}\ast \bm{u}_j) = l \ {\mbox{for
    some}} \   \bm{u}_j \in {\mathcal{U}}\nonumber \\
&& {\mbox{\ \hspace{3.5cm} where \ }} 
(i,j)  {\mbox{\  is known to be WB}}\} \nonumber 
\end{eqnarray}
and we define $\tilde{\mu}_{({\mathcal{H}},{\mathcal{U}})}^{\textrm{WB}}(i)$ by a
similar correspondence as in Lemma~\ref{lemlom1}. Let
$\bm{r}=\bm{c}+\bm{e}$ be received where $\bm{c} \in
C({\mathcal{G}},I)=C^\perp({\mathcal{H}},\widebar{I})$.
Under the assumption
$$w_H(\bm{e}) \leq \left( \min
  \{\tilde{\sigma}^{\textrm{WB}}_{\mathcal{B}}(i) \mid i \in
  I\}-1\right)/2$$
which is equivalent to
\begin{eqnarray}w_H(\bm{e}) \leq \left( \min
  \{\tilde{\mu}^{\textrm{WB}}_{\mathcal{B}}(l) \mid l \notin
  \widebar{I}\}-1\right)/2 \label{eqihop2}  \end{eqnarray}
we shall determine 
the syndromes 
$${s}_1=\bm{h}_1\cdot \bm{e} , \ldots , {s}_n=
\bm{h}_n\cdot \bm{e} .$$
From this we will be able to calculate
$$\bm{e}^T=H^{-1}\left( \begin{array}{c}s_1\\ \vdots \\
    s_n\end{array}\right).$$
Clearly, $s_i$ is known for all $i \in \widebar{I}$ as for such indexes
$\bm{h}_i\cdot \bm{e} =\bm{h}_i\cdot \bm{r} $. We next determine
the missing $s_i$'s one by one by applying the following procedure
iteratively. We shall need the notation
$$s_{vw}=(\bm{h}_v \ast \bm{u}_w)\cdot \bm{e}, {\mbox \ } 1 \leq
v,w \leq n$$
$$S(i,j)=\left(s_{vw} \right)_{\tiny{\begin{array}{c} 1 \leq v \leq i\\
1 \leq w \leq j
\end{array}}}.$$
Let $l$ be the smallest index such that $s_l$ is not
known. Consider all known ordered WB pairs $(i,j)$ such that
$\widebar{\rho}_{\mathcal{H}}(\bm{h}_i \ast \bm{u}_j)=l$. The WB property
ensures that we know $S(i-1,j-1), S(i,j-1)$ and $S(i-1,j)$. If these
three matrices are of the same rank then  $(i,j)$ is called a
candidate. For each candidate $(i,j)$ we proceed as follows:
\begin{enumerate}
\item Determine the unique scalar $s_{ij}^\prime$ such that $(s_{i1},
  \ldots , s_{i(j-1)},s_{ij}^\prime)$ belongs to the row space of $S(i-1,j)$.
\item The WB property ensures that we can write $s_{ij}$ as a linear
  combination of $s_1, \ldots , s_l$ with the coefficient of $s_l$
  being non-zero. Assuming $s_{ij}=s^\prime_{ij}$ holds, calculate
  $s_l$. This is the vote of $(i,j)$. 
\end{enumerate}
Assumption~(\ref{eqihop2}) ensures that the majority of candidates
vote for the correct value of $s_l$. Hence, we apply majority voting
and proceed to the next unknown $s_l$ if any.

In practice we will of course compute the basis ${\mathcal{H}}$ only
once and then use this basis for every received word. Doing so however
does not affect the fact that the complexity estimate of the algorithm is still
${\mathcal{O}}(n^3)$.

\begin{example}\label{exdecoding}
This is a continuation of Example~\ref{ex2}. In this example we
found the $\widebar{\sigma}^{\textrm{WB}}_{({\mathcal{G}},{\mathcal{G}})}(i)$
values. Regarding the 
$\widebar{\mu}^{\textrm{WB}}_{({\mathcal{H}},{\mathcal{G}})}(l)$
values we only noted that they could most easily be found by using the
correspondence in Lemma~\ref{lemlom1}. Doing this the array of information
regarding WB pairs with respect to $({\mathcal{H}},{\mathcal{G}})$
turns out to be exactly the same as that for WB pairs with respect to
$({\mathcal{G}},{\mathcal{G}})$. That is, the information is described
in (\ref{eqarray}).
Choosing $I=\{1,2,3,5\}$ we get
$C({\mathcal{G}},I)=C^\perp({\mathcal{H}},\widebar{I})$ where
$\widebar{I}=\{1,2,3,4,6\}$. This is a code with parameters
$[n,k,d]=[9,4,4]$. Consider the codeword
$$\bm{c}=4\bm{g}_1+3\bm{g}_2+2\bm{g}_3+\bm{g}_5=(0,3,1,4,3,2,3,3,3).$$
Let $\bm{e}=(0, \ldots , 0,1)$ and $\bm{r}=\bm{c}+\bm{e}$. We shall
use the majority voting algorithm to detect $\bm{e}$ using the
information that $\bm{c}\in C^\perp({\mathcal{H}},\widebar{I})$. The
starting point is to calculate
$$s_1=\bm{h}_1\cdot \bm{r} =4, s_2=\bm{h}_2\cdot \bm{r} =3, s_3= \bm{h}_3\cdot \bm{r} =3, s_4=\bm{h}_4\cdot \bm{r} =3, s_6=\bm{h}_6\cdot \bm{r} =3.$$
We start by looking for $s_5=\bm{h}_5
\cdot \bm{e}$. In the following array a number $k$ in position $(i,j)$
means that $s_{ij}=k$ is known and ? means that
$\widebar{\rho}_{\mathcal{H}}(\bm{h}_i \ast \bm{g}_j)=5$.
$$
\left[
\begin{array}{rrrrr}
4&\,\,2&\,\,2&\,\,1&\,\,?\\
3&\,\,4&\,\,?\\
3&\,\,?\\
3\\
?
\end{array}
\right].
$$
The candidates are $(3,2)$ and $(2,3)$. Both vote
for $s_5=1$. \\
Knowing now $s_1, \ldots , s_6$ we turn to investigate $s_7$. In the
following array again a number $k$ in position $(i,j)$
means that $s_{ij}=k$ is known; but now ? means that
$\widebar{\rho}_{\mathcal{H}}(\bm{h}_i \ast \bm{g}_j)=7$.
$$
\left[
\begin{array}{rrrrrrr}
4&\,\,2&\,\,2&\,\,1&\,\,1&\,\,1&\,\,?\\
3&\,\,4&\,\,4&\,\,2&\,\,?\\
3&\,\,4&\,\,4&\,\,?\\
3&\,\,4&\,\,?\\
1&\,\,?\\
3\\
?
\end{array}
\right].
$$
Candidates are $(5,2), (4,3), (3,4), (2,5)$. All 
vote for $s_7=1$. \\
Turning to $s_8$ the array is
$$
\left[
\begin{array}{rrrrrrrr}
4&\,\,2&\,\,2&\,\,1&\,\,1&\,\,1&\,\,3&\,\,?\\
3&\,\,4&\,\,4&\,\,2&\,\,2&\,\,?\\
3&\,\,4&\,\,4&\,\,2&\,\,?\\
3&\,\,4&\,\,4\\
1&\,\,3&\,\,?\\
3&\,\,?\\
1\\
?
\end{array}
\right].
$$
The candidates are $(6,8)$, $(5,3)$, $(3,5)$,$(2,6)$ all of which
vote for $s_8=1$.\\
Turning finally to $s_9$ we have
$$
\left[
\begin{array}{rrrrrrrrr}
4&\,\,2&\,\,2&\,\,1&\,\,1&\,\,1&\,\,3&\,\,3&\,\,?\\
3&\,\,4&\,\,4&\,\,2&\,\,2&\,\,2&\,\,1&\,\,?\\
3&\,\,4&\,\,4&\,\,2&\,\,2&\,\,2&\,\,?\\
3&\,\,4&\,\,4&\,\,2&\,\,2&\,\,?\\
1&\,\,3&\,\,3&\,\,4&\,\,?\\
3&\,\,4&\,\,4&\, \,?\\
1&\,\,3&\,\,?\\
1&\,\,?\\
?
\end{array}
\right].
$$
Candidates are $(8,2)$, $(7,3),\ldots , (2,8)$. All 
vote for $s_9=1$. \\
Finally
$$\bm{e}^T=\left[ \bm{g}_9 \cdots \bm{g}_1\right]\left(\begin{array}{c}
4\\3\\3\\3\\1\\3\\1\\1\\1\end{array}\right)=\left(\begin{array}{c}
0\\0\\0\\0\\0\\0\\0\\0\\1\end{array}\right)$$
as expected.
\end{example}

\section{Primary one-point algebraic geometric codes}\label{secfivefive}
In this section we are concerned with primary one-point algebraic
geometric codes and their improvements. We shall describe in detail the connection between
the order bound (Section~\ref{secfive}) for such codes 
and an almost similar bound in~\cite{MM}.

Following~\cite{MM} let 
$$\{\beta(1), \ldots , \beta(n)\}=\{ m \mid C_\Omega (D,mQ) \neq
C_\Omega (D,(m+1)Q)\}$$
with $\beta(1) > \cdots > \beta(n)$ and choose $\omega_i \in
\Omega(\beta(i)Q-D)$ such that $\nu_Q(\omega_i)=\beta(i)$, $i=1, \ldots ,
n$. Let
$$\{\alpha(1), \ldots ,\alpha(n)\}=\{m \mid C_{\mathcal{L}}(D,mQ) \neq
C_{\mathcal{L}}(D,(m-1)Q)\}$$
with $\alpha(1) < \cdots < \alpha(n)$ and choose $f_i \in
{\mathcal{L}}(\alpha(i)Q)\backslash {\mathcal{L}}((\alpha(i)-1)Q)$,
$i=1, \ldots , n$. This is consistent with the
notation in Definition~\ref{ovenfor}. We define
$$\bm{h}_i=({\textrm{res}}_{P_1}(\omega_i), \ldots
,{\textrm{res}}_{P_n}(\omega_i))$$
and ${\mathcal{H}}=\{\bm{h}_1, \ldots ,  \bm{h}_n \}$. Similarly we
define
$$\bm{g}_i=(f_i(P_1), \ldots , f_i(P_n))$$ and
${\mathcal{G}}=\{\bm{g}_1, \ldots , \bm{g}_n\}$. In this section we
shall always assume that ${\mathcal{H}}$ and ${\mathcal{G}}$ are
chosen as above. A standard result
tells us that 
\begin{equation}
C_{\mathcal{L}}(D,mQ)=C_\Omega(D,mQ)^\perp=C({\mathcal{G}},\widebar{I})=C^\perp({\mathcal{H}},I), \label{eqknox}
\end{equation}
where $I=\{1, \ldots , n-k\}$ and $\widebar{I}=\{1, \ldots , k\}$, for some $k$. Matsumoto
and Miura showed how to derive information on WB pairs with respect to
$({\mathcal{H}},{\mathcal{G}})$.
\begin{proposition}\label{prohallo}
Let $({\mathcal{H}},{\mathcal{G}})$ be as above. If $\beta(i)-\alpha(j)=\beta(k)$ then $(i,j)$ is WB with respect to
$({\mathcal{H}},{\mathcal{G}})$ and $\widebar{\rho}_{\mathcal{H}}(\bm{h}_i
\ast \bm{g}_j)=k$.
\qed
\end{proposition}
We now derive an alternative formulation:
\begin{proposition}\label{prolovende2}
Let $({\mathcal{H}},{\mathcal{G}})$ be as above.
If $\alpha(i)+\alpha(j)=\alpha(k)$ then $(n-k+1,j)$ is WB with respect
to $({\mathcal{H}},{\mathcal{G}})$ 
and $\widebar{\rho}_{\mathcal{H}}(\bm{h}_{n-k+1} \ast \bm{g}_j)=n-i+1$. 
\end{proposition}
\begin{proof}
We first observe that $\beta(i)=\alpha(n-i+1)-1$ holds. Rearranging
the condition in Proposition~\ref{prohallo} and
substituting $i^\prime=n-k+1$ and $k^\prime=n-i+1$ we arrive at the
proposition, except we have $i^\prime$ and $k^\prime$ where the
proposition has $i$ and $k$.
\qed
\end{proof}

Proposition~\ref{prolovende2} suggests the following
equivalent
formulation of
Matsumoto and Miura's bound on codes from algebraic function fields of
transcendence degree one \cite[Secs.\ 3 \& 4]{MM}.
Observe that differentials and residues
are completely removed from the proof, while such removal
had been completely unimaginable to the second author.

\begin{proposition}\label{pronoget}
Let $({\mathcal{H}},{\mathcal{G}})$ be as above,
and $I$ an arbitrarily chosen subset of $\{1$, \ldots, $n\}$.
The minimum distance of $C^\perp({\mathcal{H}},I)$ is at least
$$\min \{ \sigma (\alpha(i)) \mid i \in \widebar{I} \},$$
where $\sigma$ is as in Definition~\ref{defmusique}.
In particular the minimum distance of the one-point algebraic
geometric code in~(\ref{eqknox}) is at least
$$\min \{ \sigma (\alpha(i)) \mid i=1, \ldots , k\}.$$
\end{proposition}
\begin{proof}
Observe that
$$i \in \widebar{I} \Leftrightarrow n-i+1 \in \{1, \ldots , n\} \backslash
I.$$ We now apply Proposition~\ref{prolovende2} and
(\ref{bounddual3}). \qed
\end{proof}

Proposition~\ref{pronoget} and Theorem~\ref{theorderbounds}
produce the same bound for the codes in~(\ref{eqknox}). Note that
$({\mathcal{H}},{\mathcal{G}})$ satisfies the condition
in~(\ref{eq77}), but in general not the condition
in~(\ref{dualbases}). Therefore by Remark~\ref{rem77},
$C({\mathcal{G}},\widebar{I})$ is not always equal to $C^\perp
({\mathcal{H}},I)$. However, the order bound (Theorem \ref{theorderbounds})
for primary codes and
Proposition~\ref{pronoget} produce the same estimates on the minimum
distances of the two codes, and their dimensions are
the same.

\section{Estimation of generalized Hamming weights}\label{secsix}
As explained in Section~\ref{secthree} the correspondences in Lemma~\ref{lemlom1} and \ref{lemlom2} imply
that the Feng-Rao estimates  for the minimum distances of $C({\mathcal{G}},I)$ and
$C^\perp({\mathcal{H}},\widebar{I})$ are the same (although as
demonstrated in Section~\ref{secfour}, they may be
easier to derive from one of the code descriptions than the other). This is
Theorem~\ref{kjlj}. A similar result holds regarding generalized
Hamming weights which will be clear once we have stated the Feng-Rao
bounds for these parameters. Recall that the support of a set $S
\subseteq {\mathbb{F}}_{q}^{n}$ is defined by
$$\Supp(S)= \{ i \mid c_i \neq 0 {\mbox{ \ for some \ }} \bm{c}=(c_1, \ldots ,c_n)
\in S\}.$$
The $t$th generalized Hamming weight of a code $C$ is defined by
$$d_t(C)=\min\{ \sharp \Supp(S) \mid S {\mbox { is a linear subcode of \
  }} C {\mbox{ \ of dimension \ }} t \}.$$
Clearly $d_1(C)$ is just the well-known minimum distance. To introduce
the Feng-Rao bounds we will need the following two definitions. The latter
is a 
slight modification of~\cite[Def.~ 6]{geithom}.

\begin{definition}\label{defuv}
Given bases ${\mathcal{B}}, {\mathcal{U}}$ consider for $l=1, 2, \ldots ,n$ the following sets
\begin{eqnarray}
{\mathrm{V}}_{({\mathcal{B}},{\mathcal{U}})}^{{\textrm{WB}}}(l)&=&\{i \in \{1, \ldots , n\} \mid  \widebar{\rho}_{\mathcal{B}}(\bm{b}_i \ast \bm{u}_j)=l {\mbox{ \ for some \ }}
\bm{u}_j \in {\mathcal{U}}   {\mbox{\ with \ }}  (i,j) {\mbox{\   WB
}}\}, \nonumber \\
\Lambda_{({\mathcal{B}},{\mathcal{U}})}^{{\textrm{WB}}}(i)&=&\{ l \in \{1, \ldots , n\} \mid \widebar{\rho}_{\mathcal{B}}(\bm{b}_{i}\ast \bm{u}_j) = l \ {\mbox{for
    some}} \   \bm{u}_j \in {\mathcal{U}} {\mbox{\ with \ }} \ (i,j)
{\mbox{\  WB}}\}. \nonumber 
\end{eqnarray}
\end{definition}

\begin{definition}\label{defmusic}
For $1 \leq l_1 < \cdots < l_t \leq n$ define
\begin{eqnarray}
\widebar{\mu}_{({\mathcal{B}},{\mathcal{U}})}^{{\textrm{WB}}} (l_1, \ldots l_t)&=&\sharp \left( \cup_{s=1, \ldots ,
  t}{\mathrm{V}}_{({\mathcal{B}},{\mathcal{U}})}^{{\textrm{WB}}} ({l_s})\right)  \nonumber 
\end{eqnarray}
and for $1 \leq i_1< \cdots < i_t
\leq n$
 define
\begin{eqnarray}
\widebar{\sigma}_{({\mathcal{B}},{\mathcal{U}})}^{{\textrm{WB}}} 
( i_1, \ldots,  i_t
)&=&\sharp \left( \cup_{s=1,     \ldots , t}  \Lambda_{({\mathcal{B}},{\mathcal{U}})}^{{\textrm{WB}}}  ({i_s})\right). \nonumber
\end{eqnarray} 
\end{definition}

${\mathrm{V}}_{({\mathcal{B}},{\mathcal{U}})}^{{\textrm{OWB}}}$,  
$\Lambda_{({\mathcal{B}},{\mathcal{U}})}^{{\textrm{OWB}}}$,
$\widebar{\mu}_{({\mathcal{B}},{\mathcal{U}})}^{\textrm{WWB}}$ and  $\widebar{\sigma}_{({\mathcal{B}},{\mathcal{U}})}^{\textrm{WWB}}$ are defined similarly, but with WB
replaced with OWB. The Feng-Rao bounds for generalized Hamming weights are as follows:

\begin{theorem}\label{egensaet1}
Let $I$ be fixed. \\
For 
$1 \leq t \leq n- \sharp I$ we have
\begin{eqnarray}
&&d_t(C^{\perp}({\mathcal{B}},I))\nonumber \\
&& {\mbox{ \ \ \ \ }} \geq 
\min\{\widebar{\mu}_{{{({\mathcal{B}},{\mathcal{U}})}}}^{{\textrm{OWB}}}(l_1,
\ldots, l_t) \mid 1 \leq l_1 < \cdots < l_t \leq n {\mbox{ \ and \ }}
 l_1, \ldots , l_t \notin I \}\nonumber \\
&& {\mbox{ \ \ \ \ }} \geq  \min\{\widebar{\mu}_{({\mathcal{B}},{\mathcal{U}})}^{{\textrm{WB}}}(l_1,
\ldots, l_t) \mid 1 \leq l_1 < \cdots < l_t \leq n {\mbox{\  and \ }}
 l_1, \ldots , l_t \notin I \}.\nonumber 
\end{eqnarray}
For $1 \leq t \leq \sharp I$ we have
\begin{eqnarray}
&&d_t(C({\mathcal{B}},I))\nonumber \\
&&{\mbox{ \ \ \ \ }} \geq
\min\{\widebar{\sigma}_{({\mathcal{B}},{\mathcal{U}})}^{{\textrm{OWB}}}(i_1, \ldots,
i_t) \mid 1 \leq i_1 < \cdots < i_t \leq n
\ {\mbox{and}} \  i_1, \ldots , i_t  \in I \} \} \nonumber  \\
&&{\mbox{ \ \ \ \ }} \geq \min\{\widebar{\sigma}_{({\mathcal{B}},{\mathcal{U}})}^{{\textrm{WB}}}(i_1,
\ldots, i_t) \mid 1 \leq i_1<\cdots < i_t \leq n
\ {\mbox{and}} \  i_1, \ldots , i_t \in I \} \}. \nonumber  
\end{eqnarray}  
\end{theorem}
\begin{proof}
See \cite[Th.~3.14]{heijnenpellikaan} and \cite[Th.~1]{geithom}.
\qed
\end{proof}

Applying Lemma~\ref{lemlom1} and \ref{lemlom2} we get:

\begin{theorem}
Assume ${\mathcal{G}},{\mathcal{H}}$ satisfy
condition~(\ref{dualbases}) and that ${\mathcal{U}}$ is any basis. Let a non-empty set $I \subseteq \{1, 2,
\ldots , n\}$ be given. Then 
$C({\mathcal{G}},I)=C^\perp({\mathcal{H}},\widebar{I})$ and for $t \leq \sharp
I$ we have
\begin{eqnarray}
\min \{ \widebar{\mu}_{({\mathcal{H}},{\mathcal{U}})}^{\textrm{OWB}}(l_1,\ldots , l_t)
\mid 1 \leq l_1 < \cdots < l_t \leq n {\mbox{ \ and \ }} l_1, \ldots
, l_t \notin \widebar{I} \} {\mbox{ \ \ \ \ }}\nonumber \\
=\min\{\widebar{\sigma}_{({\mathcal{G}},{\mathcal{U}})}^{\textrm{OWB}}(i_1, \ldots , i_t)
\mid  1 \leq i_1<\cdots < i_t \leq n  {\mbox{ \ and \ }} i_1, \ldots ,i_t \in
I\}, \nonumber 
\end{eqnarray}
\begin{eqnarray}
\min \{ \widebar{\mu}_{({\mathcal{H}},{\mathcal{U}})}^{\textrm{WB}}(l_1,\ldots , l_t)
\mid 1 \leq l_1 < \cdots < l_t \leq n {\mbox{ \ and \ }} l_1, \ldots
, l_t \notin
\widebar{I} \} {\mbox{ \ \ \ \ }}\nonumber \\
=\min\{\widebar{\sigma}_{({\mathcal{G}},{\mathcal{U}})}^{\textrm{WB}}(i_1, \ldots , i_t)
\mid  1 \leq i_1<\cdots < i_t \leq n  {\mbox{ \ and \ }}i_1, \ldots ,i_t \in I\}. \nonumber 
\end{eqnarray}
\qed
\end{theorem}

In other words, also in the setting of generalized Hamming weights the Feng-Rao
bound for primary codes and the Feng-Rao bound for dual codes are
consequences of each other.

\section*{Acknowledgments}
The present work was done while Ryutaroh Matsumoto was visiting Aalborg
University as a Velux Visiting Professor supported by the Villum
Foundation. The authors gratefully acknowledge this support. The authors also gratefully acknowledge the support from
the Danish National Research Foundation and the National Science
Foundation of China (Grant No.\ 11061130539) for the Danish-Chinese
Center for Applications of Algebraic Geometry in Coding Theory and
Cryptography. Furthermore the authors are thankful for the support from Spanish
grant MTM2007-64704 and for 
 the MEXT Grant-in-Aid for Scientific Research (A) No.\ 23246071.


\end{document}